\documentclass{article}
\pdfoutput=1
\usepackage{amsfonts,amsmath,amssymb,amsthm}
\usepackage{verbatim,float,url,enumerate}
\usepackage{graphicx,subfigure,epsfig,psfrag}
\usepackage{bm,dsfont,color,appendix}
\usepackage{natbib}
\usepackage{latexsym,graphicx}
\usepackage{setspace}
\usepackage{adjustbox}
\usepackage[margin = 1in]{geometry}

\floatstyle{ruled}
\newfloat{algorithm}{tbhp}{loa}
\floatname{algorithm}{Algorithm}

\newtheorem{theorem}{Theorem}[section]
\newtheorem{lemma}[theorem]{Lemma}
\newtheorem{remark}[theorem]{Remark}

\newtheorem{assumptions}[theorem]{Assumption}

\newcommand{\argmin}{\mathop{\mathrm{argmin}}}

\begin{document}
\title{Test Error Estimation after Model Selection Using Validation Error}
\author{Leying Guan \thanks{Dept. of  Statistics,
    Stanford Univ,  lguan@stanford.edu}}
\maketitle
\begin{abstract}
When performing supervised learning  with the model selected using validation error from sample splitting and cross validation,  the minimum value of the validation error can be biased downward.  We propose two simple methods that use the errors produced in the validating step to estimate the test error after model selection, and we focus on the situations where we select the model by minimizing the validation error and the randomized validation error. Our methods do not require model refitting,  and the additional computational cost is negligible. 

In the setting of sample splitting, we show that, the proposed test error estimates have biases  of size $o(1/\sqrt{n})$ under suitable assumptions.  We also propose to use  the bootstrap to construct confidence intervals for the test error based on this result.  We apply our proposed methods to a number of simulations and examine their performance.
\end{abstract}

\section{Introduction}
Sample splitting and cross-validation(CV) are widely used in machine learning to choose the value of tuning parameters  in a prediction model. By training and testing the model on separate subsets of data, we get an idea of the model's prediction strength as a function of the tuning parameter, and we choose the parameter to minimize the validation error.

 Although such validation errors(sample splitting validation error and cross-validation error) are unbiased of the test error, the action of selection, for example, selecting the model with minimum validation error,  can cause the nominal validation error for the selected model to be optimistic. This can be problematic if the bias is significant and we need to rely on it to make decisions.  For example,  suppose that we have two different methods $A$ and $B$. $A$ has one tuning parameter while $B$ has five. Which one should we use? If the best model we trained using $A$ has validation error $0.1$ and the best model we trained using $B$ has validation error $0.06$, should we pick $B$ over $A$? Not necessarily. The expected test error of method $B$ associated with its model picking rule can still be higher than that of  the method $A$.

While this phenomenon is ubiquitous, especially with the emerging of many new techniques requiring heavy tuning, there isn't a satisfying solution to it yet. We take CV as an example. \citet{varma2006bias} suggests the use of ``nested" cross-validation to estimate the true test error, which is less efficient because it further divides the data and it is often impractical when the problem size is large. \citet{tibshirani2009bias} proposed a bias correction for the minimum CV error in K-fold cross-validation that directly  computes errors from the individual error curves from each fold. In their simulations, they showed that such a correction can correctly remove the downward bias, though it can sometimes introduce upward bias.

In this paper, we first propose two methods for estimating test error after selecting the model using validation error from sample splitting and then extent them to the CV setting.  The first approach estimates the optimism of the validation error by contrasting, while the second relies on randomization. In the sample splitting setting, We show that  these methods have biases of size $o(1/\sqrt{n})$ under suitable assumptions.  We  can also provide confidence intervals via the bootstrap.  

The paper is organized as follows. We describe the two procedures for test error estimation with sample splitting validation error in section \ref{sec:procedure} and analyze the their biases in section \ref{sec:theorem}.  We describe how to construct the bootstrap confidence interval in section \ref{sec:conf} and extent the proposed methods to the CV setting in section \ref{sec:extension}. We provide simulation results in section \ref{sec:sim} .

\section{Test error estimation with sample splitting validation}
\label{sec:procedure}
Suppose that we have $n$ i.i.d observations $(x_i, y_i)$ from a distribution $F(x, y)$. After sample splitting,  we have a training set $\{(x_i, y_i)| i\in D_{tr}\}$ and a validation set $\{(x_i, y_i)|i\in D_{val}\}$. Let $n$ be the size of $D_{val}$.  Let $L(x, y, \theta)$ be the loss at point $(x, y)$ with parameter $\theta$. Suppose there are $m$ models in total with model indexes being $j = 1,2,\ldots,m$. The $j^{th}$ model is trained by minimizing the loss with the penalization function $g_j(\theta)$:
\begin{equation}
\hat{\theta}_j = \arg\min_{\theta}  \frac{1}{|D_{tr}|}\sum_{i\in D_{tr}}L(x_i, y_i, \theta) + g_j(\theta)
\end{equation}

Let $L_j(x, y)$ be the loss at point $(x, y)$ with this parameter. The estimate of the test error for the model with penalization $g_j(\theta)$ (we refer to this as ``at index $j$") is 
\begin{equation}
Q_j = \frac{1}{n}\sum_{i\in D_{val}} L_j(x_i, y_i)
\end{equation}

At index $j$, this is an unbiased estimate of the test error:
\begin{equation}
\rm{Err}_{j} := E\left [ L_j(x, y)\right], \;\;(x, y) \sim F(x, y)
\end{equation}

Based on a criterion $\mathcal{R}$, we use the validation errors  to pick an index $j$  among the $m$ candidates.  We say $j\in \mathcal{R}$ if we have  picked the index $j$  under $\mathcal{R}$. The test error of criterion $\mathcal{R}$ is defined as
\begin{equation}
\rm{Err}(\mathcal{R}) := E\left[\sum^m_{j=1}\rm{Err}_{j}\mathbb{I}(j\in \mathcal{R})\right]
\end{equation}
In this paper, we will consider two criteria: ($\mathcal{R}$) the criterion which picks $j$ minimizing $Q_j$, and ($\widetilde{\mathcal{R}}$) the criterion which picks $j$ minimizing a randomized validation error to be defined later. Notice that we can also consider other rules,  for example, the one sigma rule (\cite{friedman2001elements}).

\subsection{Test error estimation with $\mathcal{R}$ }
Test error estimation with $\mathcal{R}$ is straightforward and it consists of three steps.

\noindent\rule{1\textwidth}{0.4pt}\\
\vspace{-.3cm}\\
\centerline{\bf Test error estimation with $\mathcal{R}$}
\begin{enumerate}
\item Input the $n\times m$ error matrix $L_j(x_i, y_i)$. Divide the validation errors into $K$ folds with equal size $\frac{n}{K}$ by column and the partition is $\cup_{k=1}^K S_k$. By default, $K = 2$ and the folds are created randomly.
\item Let $Q^k_j$ be the mean validation error using model $j$ with validation data from fold $k$:
$$Q^k_j =\frac{K}{n}\sum_{i\in S_k} L_j(x_i, y_i)$$ 
\item Let $j^*_k$ be the index minimizing $Q^k_j$, and let $j^*$ be the index minimizing $Q_j$. We propose the following bias correction for test error estimation:
$$\hat{\Delta} = \frac{1}{K\sqrt{K}}\sum^K_{k=1}\big(\frac{\sum_{l\neq k} Q^l_{j^*_k}}{K-1} - Q^k_{j^*_k}\big)$$
The estimated test error $\widehat{Q}(\mathcal{R})$ is given below:
$$\widehat{Q}(\mathcal{R}) =  Q_{j^*}+\hat{\Delta}$$
\end{enumerate}
\vspace{-.5cm}
\noindent\rule{1\textwidth}{0.4pt}

\subsection{Test error estimation with $\widetilde{\mathcal{R}}$}
Let $Q = (Q_1, \ldots ,  Q_m)$. We define two sequences of randomized pseudo errors,
\begin{align*}
& \widetilde{\rm{Q}}^{\alpha}(\epsilon, z)= Q+\frac{\epsilon}{\sqrt{n}}+\sqrt{\frac{\alpha}{n}}z,\\
&\widetilde{\rm{Q}}^{\frac{1}{\alpha}}(\epsilon, z) = Q+\frac{\epsilon}{\sqrt{n}}-\sqrt{\frac{1}{n\alpha}}z
\end{align*}
where  $\alpha$ is a small constant,  $\epsilon{\sim} \mathcal{N}(0,\sigma^2_0\bold{I})$, $z\sim N(0,\hat{\Sigma}+\sigma^2_0\bold{I})$ with $\sigma^2_0$ also being a small constant and $\hat{\Sigma}$ being an estimation of $\Sigma$, the underlying covariance structure of validation errors across different models. The criterion $\widetilde{\mathcal{R}}$ picks $j$ minimizing the randomized validation errors $\widetilde{\rm{Q}}^{\alpha}_j(\epsilon, z)$.  Test error estimation with $\widetilde{\mathcal{R}}$ is given below.

\noindent\rule{1\textwidth}{0.5pt}\\
\vspace{-.3cm}\\
\centerline{\bf Test error estimation with $\widetilde{\mathcal{R}}$}
\begin{enumerate}
\item  Input the $n\times m$ error matrix $L_j(x_i, y_i)$ and parameters $\alpha$, $\sigma^2_0$ , $H$ and $\hat{\Sigma}$. By default, we set $\alpha = 0.1$, $H = 100$. The default $\Sigma$ estimate is the sample covariance matrix:
$$\hat{\Sigma}_{j,j'} = \frac{\sum_{i\in D_{val}}(L_j(x_i, y_i) - Q_j)(L_{j'}(x_i, y_i) - Q_{j'})}{n}$$
The default $\sigma^2_0$ is set to be the smallest diagonal element of $\hat{\Sigma}$.
\item Generate  $H$ samples of the additive noise pair $(\epsilon, z)$.
\item At the $h^{th}$ round, let  $(\epsilon_h, \; z_{h})$ be the random vector generated and $j^*_h$  be the index chosen.  The proposed estimate of $\rm{Err}(\widetilde{\mathcal{R}})$ is $\widehat{Q}(\widetilde{\mathcal{R}}) $ given by

 $\widehat{Q}(\widetilde{\mathcal{R}}) = \frac{1}{H}\sum^H_{h=1} \widetilde{\rm{Q}}^{\frac{1}{\alpha}}_{j^*_h}(\epsilon_h, z_h)$
\end{enumerate}
\vspace{-.5cm}
\noindent\rule{1\textwidth}{0.4pt}

\section{Bias analysis for sample splitting validation}
\label{sec:theorem}
Throughout, we condition on $D_{tr}$ so that the training data is deterministic. We assume $m$ to be fixed while $n\rightarrow \infty$.  The training data may change with $n$ and we have in mind the setting where $D_{tr}$ is often of size $O(n)$. The multivariate CLT implies the following Lemma, whose proof we omit.

\begin{lemma}
\label{lem:normality}
Suppose $\sigma^2_j = \rm{Var}(L_j(x,y))\in (0, C)$ for a constant $C > 0$.  Let $Z_j = \sqrt{n}(Q_j - \rm{Err}_j)$, 

(a)For any $Z_{j}$ and $Z_{j'}$, if $\lim_{n\rightarrow\infty}\;E(\|Z_j - Z_{j'}\|^2_2) >  c$ for a constant $c > 0$, then  $Z_j - Z_{j'} $ is asymptotically normal with positive variance.

(b) $(Z_1, \ldots, Z_m)$ is asymptotically normal with bounded covariance $\Sigma$.
\end{lemma}
\subsection{Selection with $\mathcal{R}$}
\begin{assumptions}
\label{ass:ass1} 
The good models and bad models are separated: Let $j_0 = \arg\min_j \rm{Err}_j$, $J_{good} := \{j| \sqrt{n}(\rm{Err}_j - \rm{Err}_{j_0}) \rightarrow 0 \}$ and $J_{bad} :=\{j|\frac{\sqrt{n}}{\log n}(\rm{Err}_{j}- \rm{Err}_{j_0}) \rightarrow \infty \}$, we have
\begin{equation}
J_{good}\cup J_{bad} =\{1,2,\ldots,m\}
\end{equation}
\end{assumptions}
\begin{remark}
In practice, the differences between $\rm{Err}_1, \ldots, \rm{Err}_m$ may decrease with $n$, since it is common to use a finer grid for the tuning parameter as n grows. For example, in lasso regression, it is common to use a grid of size $\frac{1}{\sqrt{n}}$ for the penalty parameter.
\end{remark}

Let $\Delta_{out}(n)$ be the expected bias of nominal validation error after selection according to $\mathcal{R}$ with a validation set of size $n$, and let $\Delta := E(\widehat{\Delta})$, the expectation of our  bias estimate. Recall the expression of $\widehat{\Delta}$:
$$\widehat{\Delta} = \frac{1}{K\sqrt{K}}\sum^K_{k=1}\big(\frac{\sum_{l\neq k} Q^l_{j^*_k}}{K-1} - Q^k_{j^*_k}\big)$$
For a given fold $k$, $j^*_k$ is the best index selected using a validation set of size $\frac{n}{K}$, and $\frac{\sum_{l\neq k} Q^l_{j^*_k}}{K-1}$ is the unbiased error estimate using validation errors from other folds. By definition, we have
\begin{equation}
\Delta = \frac{1}{\sqrt{K}}\Delta_{out}(\frac{n}{K})
\end{equation}
 
\begin{theorem}
\label{thm:scale}
Suppose $\sigma^2_j = \rm{Var}(L_j(x,y))\in (0, C)$ for a constant $C > 0$. Then under Assumption \ref{ass:ass1}, we have that the test error $\Delta_{out}(.)$, as a function of the size of validation set, satisfies the following relationship with high probability:
\begin{equation}
\sqrt{n}\Delta_{out}(n) - \sqrt{\frac{n}{K}}\Delta_{out}(\frac{n}{K})  \rightarrow 0 
\end{equation}
As a result,
\begin{equation}
\sqrt{n}(E[\widehat{Q}(\mathcal{R})] - \rm{Err}(\mathcal{R}))\rightarrow 0
\end{equation}
\end{theorem}

\subsection{Selection with $\widetilde{\mathcal{R}}$}
Note that if $\sqrt{n}(Q-{\rm Err}) \sim N(0, \Sigma)$ and $\Sigma$ is known, then $\widetilde{Q}^{\alpha}$ and  $\widetilde{Q}^{\frac{1}{\alpha}}$ are independent by construction. Hence, $\widetilde{Q}^{\frac{1}{\alpha}}$ may be used to estimate the test error of the model selected using $\widetilde{Q}^{\alpha}$. The idea of randomized model selection has been studied before  (\cite{dwork2008differential,tian2015selective}), and this trick of constructing independent variables has been discussed  in \cite{harris2016prediction}.  By Lemma \ref{lem:normality},  we see that $\sqrt{n}(Q-{\rm Err})\sim N(0, \Sigma)$ asymptotically. This yields the following Theorem.
\begin{theorem}
\label{thm:remainder}
Suppose $\sigma^2_j = \rm{Var}(L_j(x,y))\in (0, C)$ for a constant $C > 0$ and $\|\hat{\Sigma} - \Sigma\|_{\infty}\rightarrow 0$. Then
\begin{equation}
\sqrt{n}(E[\widehat{Q}(\widetilde{\mathcal{R}})] -  \rm{Err}(\widetilde{\mathcal{R}})) \rightarrow 0
\end{equation}
\end{theorem}

The proofs of Theorem \ref{thm:scale} and Theorem \ref{thm:remainder} are given in section \ref{app:proof}.

\section{Confidence interval construction}
\label{sec:conf}
We generate $B$ bootstrap samples of $\widehat{Q}(\mathcal{R})$ and $\widehat{Q}(\widetilde{\mathcal{R}})$, and write them as $\widehat{Q}(\mathcal{R})_b$ and $\widehat{Q}(\widetilde{\mathcal{R}})_b$, for $b=1,2,\ldots, B$. Let $Q(\mathcal{R}) = \frac{\sum^B_{b=1}Q(\mathcal{R})_b}{B}$ and  $Q(\widetilde{\mathcal{R}}) = \frac{\sum^B_{b=1}Q(\widetilde{\mathcal{R}})_b}{B}$ where $Q(\mathcal{R})_b$ and $Q(\widetilde{\mathcal{R}})_b$ are the original validation errors of the selected model(s) at $b^{th}$ repetition under criterion $\mathcal{R}$ and $\widetilde{\mathcal{R}}$ respectively. We approximate the confidence intervals with coverage $(1-\alpha)$ for $\rm{Err}(\mathcal{R})$ and $\rm{Err}(\widetilde{\mathcal{R}})$ by 
\[
 [\widehat{Q}(\mathcal{R})+a_1-\frac{1}{\sqrt{n}\log n}, \widehat{Q}(\mathcal{R})+b_1+\frac{1}{\sqrt{n}\log n}]
\]
\[
 [\widehat{Q}(\widetilde{\mathcal{R}})+a_2-\frac{1}{\sqrt{n}\log n}, \widehat{Q}(\widetilde{\mathcal{R}})+b_2+\frac{1}{\sqrt{n}\log n}]
\]
where $(a_1, b_1)$, $(a_2, b_2)$  are the lower and upper $\frac{\alpha}{2}$ quantiles of the two bootstrap distributions $\{\widehat{Q}(\mathcal{R})_b - Q(\mathcal{R})\}$ and $\{\widehat{Q}(\widetilde{\mathcal{R}})_b - Q(\widetilde{\mathcal{R}})\}$.

Note that such a construction is better than a construction using the uncorrected validation error. For example, as a result of the bias-correction, the difference of the mean of $\widehat{Q}^b(\widetilde{\mathcal{R}}) -Q(\widetilde{\mathcal{R}})$ and the mean of $Q(\widetilde{\mathcal{R}}) - {\rm Err}(\widetilde{\mathcal{R}})$ is of size $o(\frac{1}{\sqrt{n}})$, while it is usually not true using uncorrected nominal validation error after selection. The bootstrap distribution of the corrected one is more similar to the true distribution in this sense.

In the confidence intervals constructed above,  we have added $\frac{1}{\sqrt{n}\log n}$ to the upper and lower boundaries of the usual Bootstrap  confidence intervals to account for the potential small order bias in finite samples. As $n\rightarrow \infty$,  this extra length will diminish compared with the total interval length.

\section{Extension to CV}
\label{sec:extension}
In CV setting, We divide the data into $K$ folds, $\cup_{k=1}^K S_k$, and perform cross validation to select a model from $m$ candidates.  The $j^{th}$ model for fold $k$ is trained by minimizing the loss with the penalization function $g_j(\theta)$, with  the $k^{th}$ fold left out:
\begin{equation}
\hat{\theta}^k_j = \argmin_{\theta}  \sum_{i\notin S_k}L(x_i, y_i, \theta) + g_j(\theta),\;\forall j = 1,2,\ldots, m
\end{equation}
Let $L^k_j(x, y)$ be the loss at point $(x, y)$ with this parameter. The validation error in the CV setting is defined as:
\[
Q_j = \frac{1}{n}\sum^K_{k=1}\sum_{i\in S_k} L^k_j(x_i, y_i)
\]
We are interested in 
\begin{align*}
&\rm{Err}(\mathcal{R}) := E\left[\sum^m_{j=1}\rm{Err}_{j}\mathbb{I}(j\in \mathcal{R})\right]\\
&\rm{Err}(\widetilde{\mathcal{R}}) := E\left[\sum^m_{j=1}\rm{Err}_{j}\mathbb{I}(j\in \mathcal{R})\right]
\end{align*}
where $\rm{Err}_{j} := E\left [ \frac{1}{K}\sum^K_{k=1} L^k_j(x, y)\right]$, with a new independent sample $(x, y) \sim F(x, y)$. For both test error estimations, we use same procedures as in section \ref{sec:procedure}, with CV error matrix $L_j^k(x_i, y_i)$ being the input error matrix. For the test error estimation with $\mathcal{R}$, instead of using the default partition in the first step, we use the partition $\cup_{k=1}^K S_k$ from CV and the We can also construct bootstrap confidence interval as described in section \ref{sec:conf} with a slight adjustment:   We sample with replacement the validation errors $L^k_j(x_i, y_i)$ within each fold to keep the structure of CV.

\section{Simulations}
\label{sec:sim}
We consider three types of error:
\begin{enumerate}
\item Under $\mathcal{R}$, the nominal validation error $Q_{j^*}$ for the selected model, denoted as A1.
\item Under $\mathcal{R}$, the bias-corrected test error estimate $\widehat{Q}(\mathcal{R}) := Q_{j^*}+\hat{\Delta}$, denoted as A2.
\item Under $\widetilde{\mathcal{R}}$, the randomized test error estimate $\widehat{Q}(\widetilde{\mathcal{R}}) := \frac{1}{H}\sum^H_{h=1}\widetilde{Q}^{\frac{1}{\alpha}}(\epsilon_h, z_h)$,  denoted as A3.
\end{enumerate}
The parameter $\alpha$ for $\widetilde{\mathcal{R}}$ is fixed at $\alpha = 0.1$, and the number of repetitions $H$ is fixed at $H= 100$. For the bootstrap confidence intervals, we set the number of bootstrap repetitions as $B = 1000$, and  fix the coverage at $90\%$. We use validation error to refer to sample splitting validation error or cross validation error according to the context.

In simulation S0, S1, we consider the sample splitting setting and we directly model the validation errors $L_j(x_i, y_i)$. \\

\noindent\textbf{S0: i.i.d validation errors(no signal):} We consider $n = 100$ and $m = 30$, and generate each $L_j(x_i, y_i)$ independently as $N(0,1)$. Here, the true test errors for all models are $0$. \medskip

\noindent\textbf{S1: i.i.d validation errors(with signal):} We consider $n = 100$ and $m = 30$. For each index $j$, we generate the true test error $\mu_j\sim N(0,\frac{1}{n})$, and generate $L_j \sim N(\mu_j, 1)$.
\medskip

 In simulations S2, S3 and S4 below, we consider classification problems and train models using 5-fold CV with the 0-1 loss function.

\medskip

\noindent\textbf{S2: Classification(no signal):} We generate $n$ i.i.d samples $(x_i, y_i)$ as $x_i \sim N(0, I_{p\times p})$, $y_i \sim \rm{Binomial}(0.5)$.  We consider the cases $(n, p) = (100, 10)$ and  $(n, p) = (100, 2000)$, and in each case, we  train models using logistic regression with lasso penalty.  
\medskip

\noindent\textbf{S3: Classification(with signal, orthogonal case):} We generate $n$ i.i.d samples $(x_i, y_i)$ as $x_i \sim N(0, I_{p\times p})$, $y_i \sim \rm{Binomial}(\frac{1}{1+e^{-x^T_i\beta}})$.  We consider the cases $(n, p) = (100, 10)$ and  $(n, p) = (100, 2000)$. In each case,  we let the first 10 features contain the signal $$\beta_j = \left\{\begin{array}{ll}4 & j\leq 10\\0& j>10\end{array}\right.$$ We train models using logistic regression with lasso penalty.
\medskip

\noindent\textbf{S4: Classification(with signal, correlated case):} The same as S3, except with $x_i \sim N(0, \Sigma)$, where $\Sigma$ is a covariance matrix with first order autoregressive structure $\Sigma_{k,k'} = 0.5^{|k-k'|}$.
\medskip

Figure \ref{fig:boxplots_errorMatrix} shows  boxplots for the differences between the estimated and true test errors in S0 and S1 across 1000 simulations. Table \ref{tab:coverage} shows the empirical coverage of the  $90\%$ bootstrap confidence intervals. 

\begin{table}[H]
\centering
\caption{Empirical coverage of the $90\%$ bootstrap confidence interval}
\vskip 0.1in
\begin{tabular}{|r|r|r|}
  \hline
 & S0 & S1 \\ 
  \hline
A2 & 0.97 & 0.91 \\ \hline
  A3 & 0.94 & 0.93 \\ 
   \hline
\end{tabular}

{\em A2,  A3 are de-biased test error estimates for $\mathcal{R}$ and $\widetilde{\mathcal{R}}$ respectively.} 
\label{tab:coverage}
\end{table}

\begin{figure}
\begin{center}
\includegraphics[width=.8\textwidth, height = .6\textwidth]{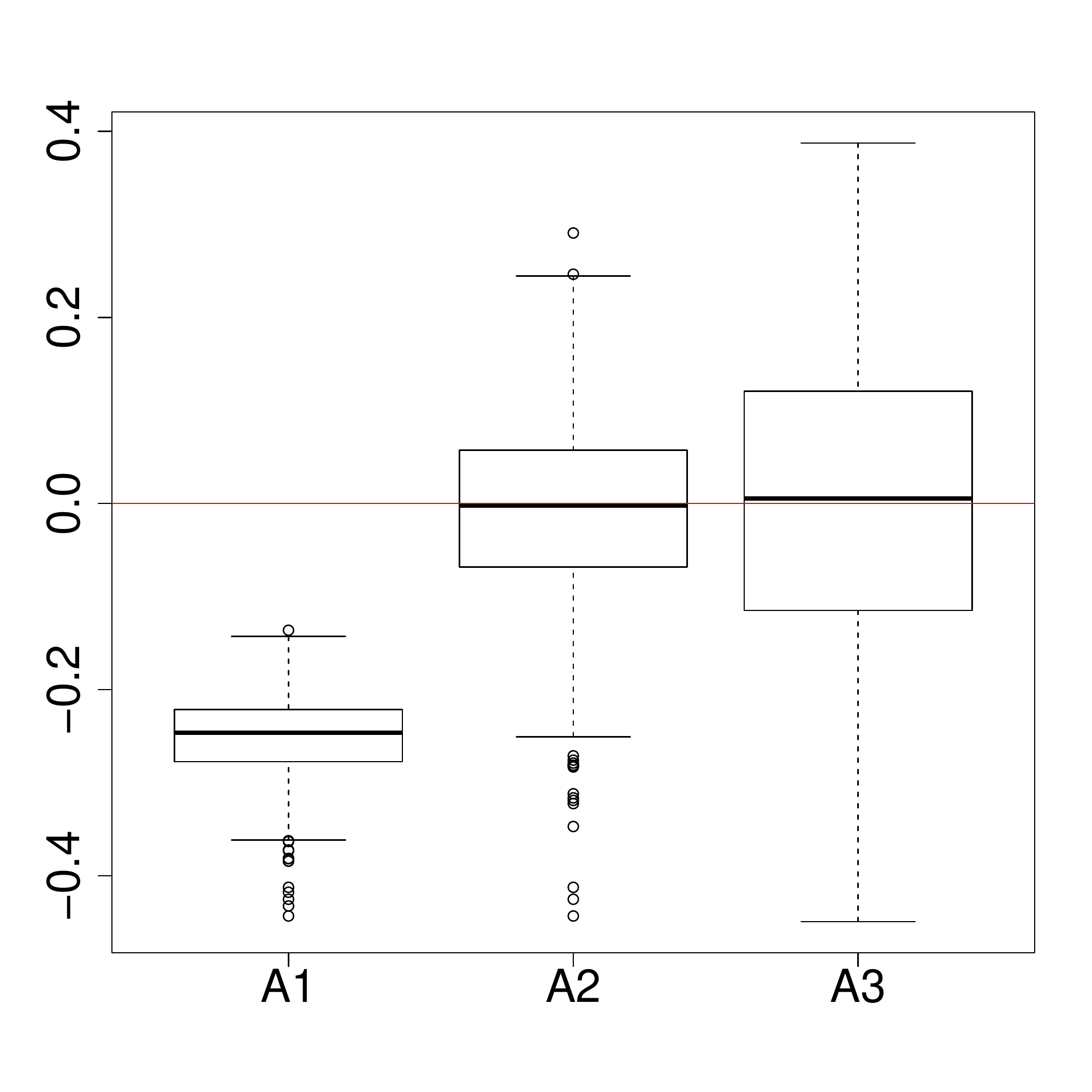}

\includegraphics[width=.8\textwidth, height = .6\textwidth]{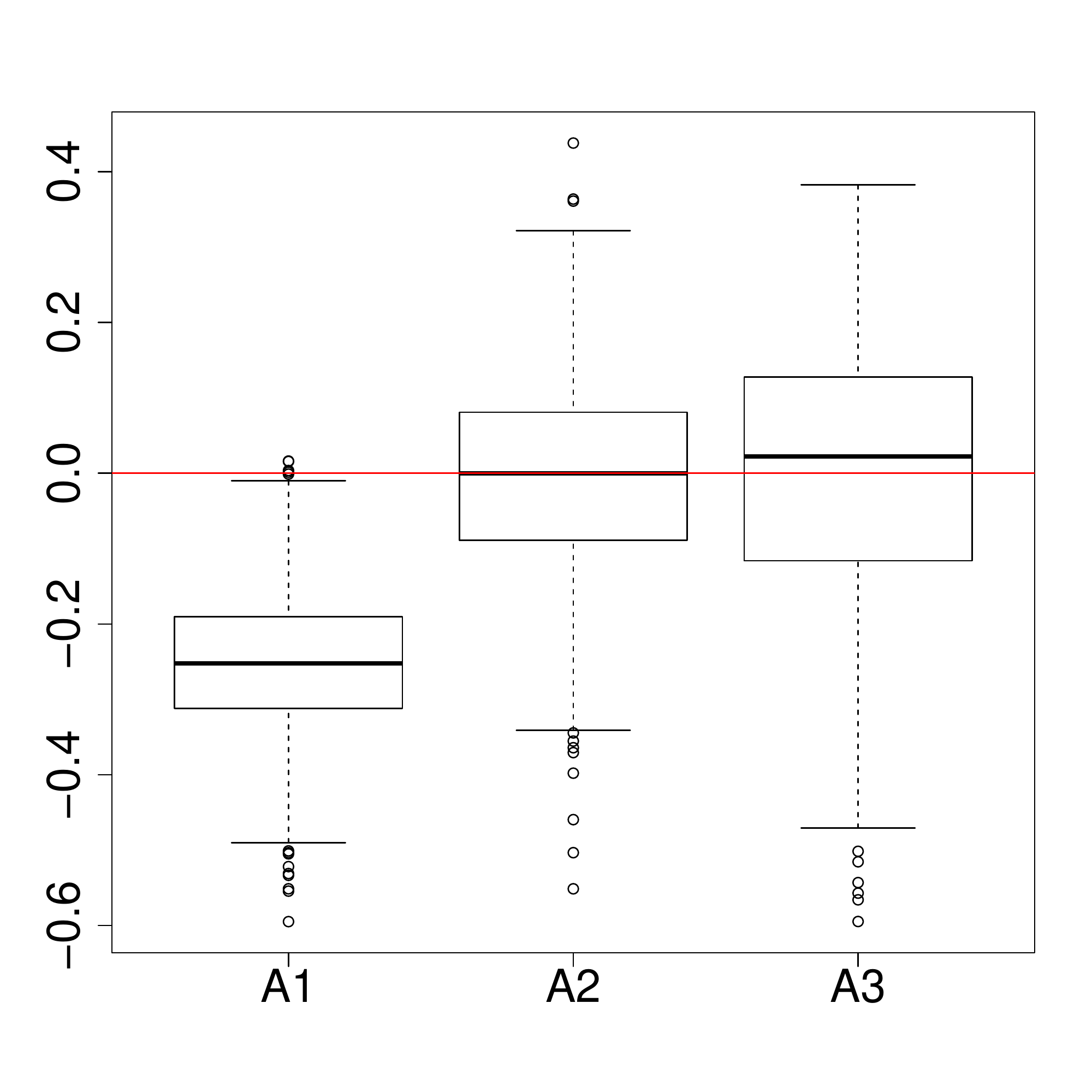}
\caption[boxplots]{\em Results for simulation experiments S0 (top) and S1 (bottom). Box-plots  show differences between the estimated and true test errors. A1, A2, A3 refer to the nominal validation error and de-biased estimates for $\mathcal{R}$ and $\widetilde{\mathcal{R}}$ respectively. We see that the nominal validation error is biased downward significantly.}
\label{fig:boxplots_errorMatrix}
\end{center}
\end{figure}

\noindent Table \ref{tab:logisitc} contains results from simulations S2, S3 and S4. The upper half shows the mean test error and our estimates of them across 1000 simulations. The lower half shows the empirical coverage of the bootstrap confidence intervals. We see that selection using the randomized criterion $\widetilde{\mathcal{R}}$ does not lead to higher test error than selection using $\mathcal{R}$. Figure \ref{fig:boxplots_logistic} shows the box-plots of the differences between estimated error and true test error in these three settings.

\begin{figure}
\begin{center}
\includegraphics[width=\textwidth, height = 1.2\textwidth]{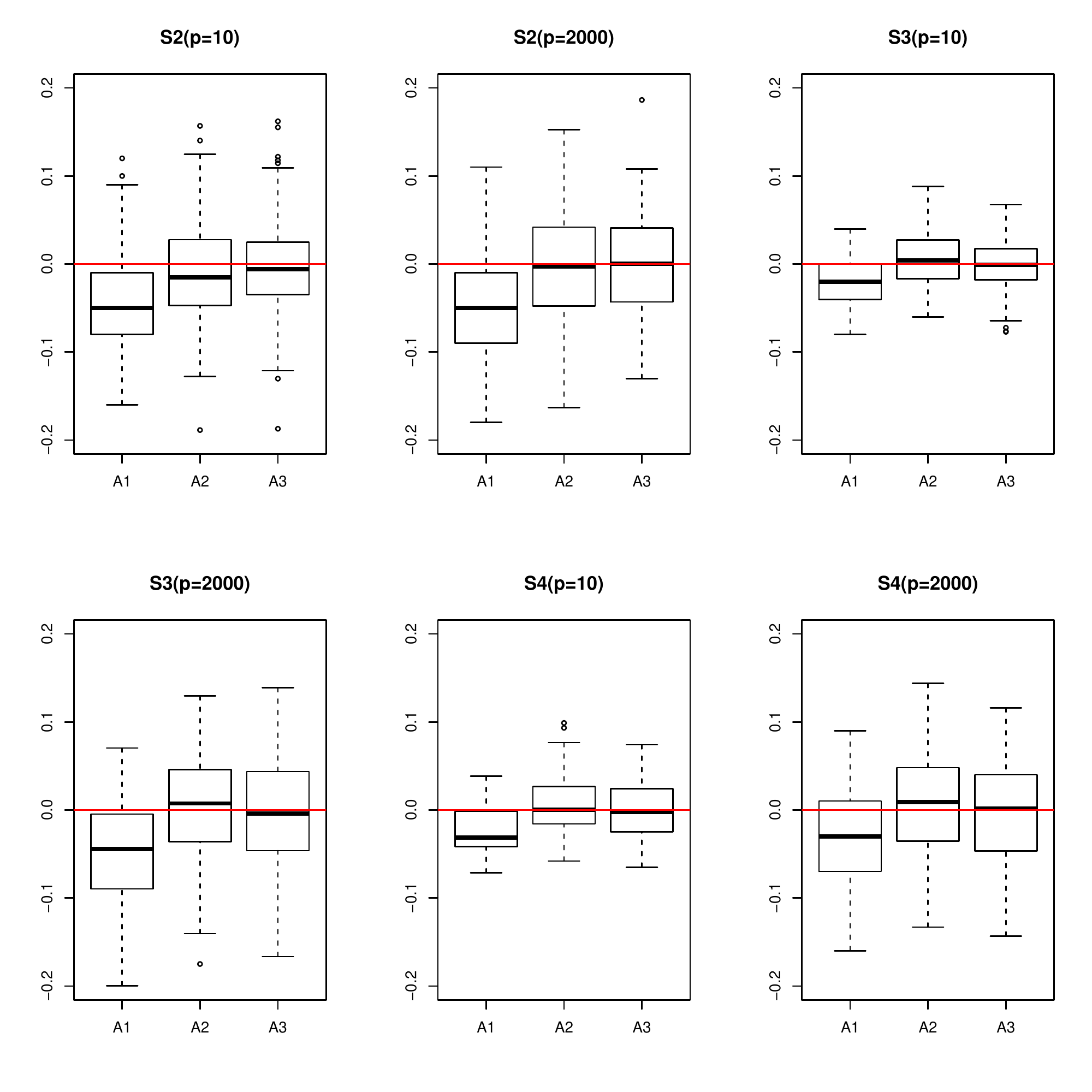}
\caption[boxplots]{\em Results for simulation experiments S2, S3, S4.  See Figure \ref{fig:boxplots_errorMatrix} for details. A2 and A3 correct the downward bias of the nominal validation error A1, which is most severe in the setting S2 with no signal.}
\label{fig:boxplots_logistic}
\end{center}
\end{figure}

\begin{table}[ht]
\centering
\caption{Test error and error estimates for S2, S3, S4}
\label{tab:logisitc}
\vskip 0.1in
\begin{adjustbox}{width=1\textwidth}
\begin{tabular}{r|llllll}
\hline
  \hline
    &&&Errors and error estimates\\ 
\hline
  \hline
 & S2(p=10) & S2(p=2000) & S3(p=10) & S3(p=2000) & S4(p=10) & S4(p=2000) \\  \hline
  Err &  0.5(0) & 0.5(0) & 0.09(0.002) & 0.408(0.004) & 0.077(0.002) & 0.192(0.003) \\   \hline
  ErrRandom &0.5(0) & 0.5(0) & 0.088(0.001) & 0.405(0.003) & 0.074(0.001) & 0.196(0.003) \\   \hline
  A1 & 0.456(0.006) & 0.445(0.006) & 0.071(0.003) & 0.359(0.007) & 0.057(0.002) & 0.163(0.005)  \\   \hline
  A2 & 0.494(0.006) & 0.491(0.006) & 0.099(0.003) & 0.406(0.007) & 0.085(0.003) & 0.202(0.005)  \\   \hline
  A3  & 0.494(0.006) & 0.494(0.006) & 0.086(0.003) & 0.402(0.007) & 0.072(0.003) & 0.196(0.005)\\   \hline
   \hline
   \hline
   &&&Coverage\\
   \hline
   \hline
    & S2(p=10) & S2(p=2000) & S3(p=10) & S3(p=2000) & S4(p=10) & S4(p=2000) \\ 
  \hline
  A2 & 0.90 & 0.89 & 0.95 & 0.89 & 0.95 & 0.91 \\ \hline
  A3 & 0.90 & 0.87 & 0.96 & 0.87 & 0.97 & 0.93 \\   \hline
  \hline
\end{tabular}
\end{adjustbox}
\flushleft
{\em The upper half shows the mean test error and error estimates across 1000 simulations, with estimated standard deviations of these mean values in parentheses. Err and ErrRandom are the test errors after selection using $\mathcal{R}$ and $\widetilde{\mathcal{R}}$.  A1, A2, A3 are as before. The lower half shows the coverage of $90\%$ bootstrap intervals.}
\end{table}

\noindent\textbf{Further simulations:} 

Here we follow the setup in \cite{tibshirani2009bias}, and we consider only the case where $p \gg n$. The features were generated as Gaussian with  $n = 40, p = 1000$. There were two classes of equal size. We created two settings: ``no signal",  in which all features were from $N(0,1)$, and ``signal", where the mean of the first 10\% of the features was shifted to be 0.2 units higher in class 2. In each of these settings we applied three different classifiers: NSC(nearest shrunken centroids),  CART (classification and regression trees), KNN ($K$-nearest neighbors), and we call this simulation S5. Table \ref{tab:moreSim} shows results of S5. Similar to the previous simulations and in \cite{tibshirani2009bias}, we see that the bias tends to larger in the``no signal" case, and varies significantly depending on the classifier. 
\begin{table}[ht]
\centering
\caption[boxplots]{\em Results for simulation experiments S5}
\label{tab:moreSim}
\vskip 0.1in
\begin{adjustbox}{width=1\textwidth}
\begin{tabular}{r|llllll}
  \hline
 & Method & A1& A2& Err & A3 & ErrRandom \\ 
  \hline
no Signal & NSC & 0.45(0.019) & 0.532(0.017) & 0.5 & 0.504(0.025) & 0.5 \\   \hline
   & CART & 0.462(0.017) & 0.502(0.014) &0.5 & 0.503(0.014) & 0.5 \\   \hline
   & KNN & 0.425(0.01) & 0.521(0.01) & 0.5 & 0.515(0.013) & 0.5 \\   \hline
  with Signal & NSC & 0.075(0.008) & 0.102(0.009) & 0.1(0.004) & 0.083(0.009) & 0.1(0.008) \\   \hline
   & CART & 0.45(0.008) & 0.488(0.009) & 0.475(0.009) & 0.474(0.008) & 0.475(0.012) \\   \hline
   & KNN & 0.125(0.01) & 0.195(0.006) & 0.188(0.013) & 0.186(0.01) & 0.194(0.01) \\ 
   \hline
\end{tabular}
\end{adjustbox}
\flushleft
{\em  See Table \ref{tab:logisitc} for details.}
\end{table}

\section{Discussion}
We have proposed two methods to estimate the error after model selection using validation error from sample splitting and extended them to the CV setting. We have seen that both bias correction approaches have reasonably good performance in our simulations.  Although one approach uses a randomized CV criterion, we do not observe deterioration of the selected model in test error. Both these methods  require no model refitting and are applicable to high dimensional cases where model fitting is computational expensive. We believe it is helpful to use the proposed error estimation approaches to guard against over-optimism when looking at the validation results after selection. An R package called debiasedCV will soon be available on the public CRAN repository.

\vskip 14pt
\noindent {\large\bf Acknowledgements} The author would like to thank Professor Robert Tibshirani, Zhou Fan and Professor Jonanthan Taylor for their suggestions and feedbacks, especially Professor Robert Tibshirani and Zhou Fan for whose comments have greatly improved the manuscript.
\par

\appendix
\section{Proof of Theorem~\ref{thm:remainder} and Theorem~\ref{thm:scale}}
\label{app:proof}
Lemma \ref{lem:thm1} below is useful in the proof of Theorem \ref{thm:remainder} and Theorem \ref{thm:scale}.
\begin{lemma}
\label{lem:thm1}
Let m be a fixed number and $\{x_n\}$, $z$ be m dimensional vectors such that $z \sim N(0, \Sigma)$ where $\Sigma$ is positive definite and $x_n\overset{D}{\rightarrow} z$, and $E[\|x_n\|^2_2]$ is asymptotically bounded.  For any bounded function $g(.)$, we have
\[
\lim_{n\rightarrow \infty} E[x_n g(x_n)] = E[zg(z)]
\]
\end{lemma}
\begin{proof}
We apply  Portmanteau Theorem and use the fact that that for all $j = 1,2,\ldots,m$, $(x_n)_j$ has asymptotically  bounded variance. $\forall \delta >0$, there exist a constant $C(\delta)$, and an index $n_0$ large enough, such that, for all $j$,
\begin{equation}
\label{eq:tailbound}
 \sqrt{E[z^2_jg^2(z)]}P(z^2_j\geq C(\delta)) < \frac{\delta}{2}, \;\;\forall n>n_0, \sqrt{E[(x_n)^2_jg^2(x_n)]}P((x_n)^2_j\geq C(\delta)) < \frac{\delta}{2}
\end{equation}
Use equation (\ref{eq:tailbound}) and Portmanteau Theorem again:
\begin{align*}
&\lim_{n\rightarrow \infty}\|E[x_ng(x_n)] - E[zg(z)]\|_\infty \\
\leq&\lim_{n\rightarrow \infty}\max_j |E[(x_n)_jg(x_n)\mathbb{I}[(x_n)^2_j<C(\delta)]] - E[z_jg(z)\mathbb{I}[z^2_j<C(\delta)]]|\\
 +&\lim_{n\rightarrow \infty}\max_j |E[(x_n)_jg(x_n)\mathbb{I}[(x_n)^2_j\geq C(\delta)]]|+\max_j|E[z_jg(z)\mathbb{I}[z^2_j\geq C(\delta)]]|\\
=&\lim_{n\rightarrow \infty}\max_j |E[(x_n)_jg(x_n)\mathbb{I}[(x_n)^2_j\geq C(\delta)]]|+\max_j|E[z_jg(z)\mathbb{I}[(x_n)^2_j\geq C(\delta)]]|\\
\leq& \lim_{n\rightarrow \infty}\max_j \sqrt{E[(x_n)^2_jg^2(x_n)]}P((x_n)^2_j\geq C(\delta))+\max_j\sqrt{E[z^2_jg^2(z)]}P(z^2_j\geq C(\delta)) < \delta\\
&\Rightarrow \lim_{n\rightarrow \infty}\|E[x_ng(x_n)] - E[zg(z)]\|_\infty = 0
\end{align*}

\end{proof}

\subsection*{Proof of Theorem \ref{thm:remainder}}
\begin{proof}
Let $t^{\alpha}_j = Q_j - Err_j+\frac{\epsilon_j}{\sqrt{n}}+\sqrt{\frac{\alpha}{n}}z_j$, $t^{\frac{1}{\alpha}} =    Q_j - Err_j+\frac{\epsilon_j}{\sqrt{n}}-\sqrt{\frac{1}{n\alpha}}z_j$, we have that $(t^{\alpha}, t^{\frac{1}{\alpha}})$ is asymptotically multivariate Gaussian.
\begin{equation}
\label{eq:CLT}
\sqrt{n}\left(\begin{array}{l} t^{\alpha}\\ t^{\frac{1}{\alpha}} \end{array}\right)\overset{d}{\rightarrow}\mathcal{N}(0,\left(\begin{array}{ll}\Sigma+\alpha\Sigma+\sigma^2_0I(1+\alpha)&\bold{0}\\ \bold{0}& \Sigma+\frac{1}{\alpha}\Sigma+\sigma^2_0I(1+\frac{1}{\alpha})\end{array}\right))
\end{equation}
Let $\left(\begin{array}{l} u^{\alpha}\\ u^{\frac{1}{\alpha}} \end{array}\right)$ be a multivariate gaussian vector generated from this limiting distribution.

We define $T_j, \;U_j$  as follows
\begin{equation}
T_j =  \{\sqrt{n}t^{\alpha}_j +\sqrt{n}\rm{Err}_j<\sqrt{n}t^{\alpha}_{j'} +\sqrt{n}\rm{Err}_{j'}, \forall j'\neq j\}
\end{equation}
\begin{equation}
U_j = \{u^{\alpha}_j +\sqrt{n}\rm{Err}_j<u^{\alpha}_{j'} +\sqrt{n}\rm{Err}_{j'}, \forall j'\neq j\}
\end{equation}
We have
\begin{align*}
& \lim_{n\rightarrow \infty} n^{\frac{1}{2}}|E[\widehat{Q}(\widetilde{\mathcal{R}})] -  Err(\widetilde{\mathcal{R}})| = \lim_{n\rightarrow \infty}n^{\frac{1}{2}}|E[\sum^m_{j=1} t^{\frac{1}{\alpha}}_j\rm{I}[T_j]]\overset{(a)}{=}  \lim_{n\rightarrow \infty}|E[\sum^m_{j=1} u^{\frac{1}{\alpha}}_j\rm{I}[U_j]] \overset{(b)}{=} 0
\end{align*}
The equation (a) is a direct result applying Lemma \ref{lem:thm1}, and the equation (b) uses the independence between $u^{\frac{1}{\alpha}}$ and $u^{\alpha}$.

\end{proof}
\subsection*{Proof of Theorem \ref{thm:scale}}
\begin{proof}
Let  $Q_{j, 1}$ be the validation error with validation set size $n$, $Q_{j, 2}$ be the validation error with validation set size $\frac{n}{K}$. Let  $Z_{j,1} = \sqrt{n}(Q_{j, 1} - {\rm Err}_j)$,  $Z_{j,2} = \sqrt{\frac{n}{K}}(Q_{j, 2} -  {\rm Err}_j)$,   $U_1(Z, j) := \{Z_{j}+\sqrt{n}{\rm Err}_j < Z_{j'}+\sqrt{n}{\rm Err}_{j'}, \forall j'\neq j\}$ and $U_2(Z, j) := \{Z_{j}+\sqrt{\frac{n}{K}} {\rm Err}_j < Z_{j'}+\sqrt{\frac{n}{K}} {\rm Err}_{j'}, \forall j'\neq j\}$, and $U_1(Z_1,j)$ is the event of selecting index $j$ using the first validation set while $U_2(Z_2, j)$ is the  event of selecting index $j$ using the second validation set. Under Assumption \ref{ass:ass1} and use the tail bound for normal distribution, we have that $\forall j_1\in J_{bad},\; j_2\in J_{good}$, 
\begin{equation}
\label{eq:tailBound0}
\lim_{n\rightarrow 0} P(Z_{j_1}+\sqrt{n} {\rm Err}_{j_1} \leq Z_{j_2}+\sqrt{n} {\rm Err}_{j_2}) = 0
\end{equation}
From Lemma \ref{lem:normality}, there exists an index set $J$, such that for any two indexes  $j$, $j'$, we have the asymptotic joint normality of $(Z_j - Z_{j'})$, and for all index $j'$, there exists an index $j\in J$, such that $\lim_{n\rightarrow \infty} E(Z_j - Z_j')^2 = 0$ and define $A_j :=\{j'|\lim_{n\rightarrow \infty} E(Z_j - Z_j')^2 = 0\}$. 
We know that $Z_{1}$ and $Z_{2}$ converge to the same Gaussian distribution, and let $Z$ be the varianble generated from it. Applying Lemma \ref{lem:thm1}, we have

\begin{align*}
&|\sqrt{n}\Delta_{out}(n) - \sqrt{\frac{n}{K}}\Delta_{out}(\frac{n}{K})|\\
=&|E[\sum^m_{j=1}Z_{j,1}\mathbb{I}_{U_1(Z_1, j)}] - E[\sum^m_{j=1}Z_{j,2}\mathbb{I}_{U_2(Z_2, j)}]|\\
=&|E[\sum_{j\in J}\sum_{j'\in A_j}Z_{j'}(\mathbb{I}_{U_1(Z, j')}-\mathbb{I}_{U_2(Z, j')})]|\\
\leq &|E[\sum_{j\in J}Z_{j}\sum_{j'\in A_j}(\mathbb{I}_{U_1(Z, j')}-\mathbb{I}_{U_2(Z, j')})]|+\sum_{j\in J}\sum_{j'\in A_j}\sqrt{E[(Z_{j'}-Z_{j})^2]}(P(U_1(Z, j')+P(U_2(Z, j'))\\
\leq &|E[\sum_{j\in J}Z_{j}\sum_{j'\in A_j}(\mathbb{I}_{U_1(Z, j')}-\mathbb{I}_{U_2(Z, j')})]|+2\max_{j\in J, j'\in A_j}\sqrt{E[(Z_{j'}-Z_{j})^2]}\\
\end{align*}

The second term goes to $0$ by definition. For the first term, we have
\begin{align*}
&|E[\sum_{j\in J}Z_{j}\sum_{j'\in A_j}(\mathbb{I}_{U_1(Z, j')}-\mathbb{I}_{U_2(Z, j')})]|\\
\leq & \max_j \sqrt{E[Z^2_{j}]}|P(\cup_{j'\in A_j}U_1(Z, j')/\cup_{j'\in A_j}U_2(Z, j'))|+\max_j \sqrt{E[Z^2_{j}]}|P(\cup_{j'\in A_j}U_2(Z, j')/\cup_{j'\in A_j}U_1(Z, j'))|
\end{align*}
 and we only need to show that, $\forall j\in J$
\begin{align*}
&|P(\cup_{j'\in A_j}U_1(Z, j')/\cup_{j'\in A_j}U_2(Z, j'))|\rightarrow 0\\
&|P(\cup_{j'\in A_j}U_2(Z, j')/\cup_{j'\in A_j}U_1(Z, j'))| \rightarrow 0
\end{align*}
Note that 
\begin{align*}
&\cup_{j'\in A_j}U_1(Z, j')/\cup_{j'\in A_j}U_2(Z, j')  \\
=&\{\exists j'\in A_j, \; Z_j' < \min_{l\notin A_j} Z_l + \sqrt{n}({\rm Err}_l - {\rm Err}_{j'}), \forall j'\in A_j,  Z_{j'} \geq  \min_{l\notin A_j} Z_l + \sqrt{\frac{n}{K}}({\rm Err}_l - {\rm Err}_{j'})\}
\end{align*}
There are only three different situations. And it is easy to check that  the above event happens with probability  goes to 0 in all of them.

(1) If  $A_j\cap J_{good} = \emptyset $. Let $j_0$ be the index corresponding to the smallest test error, we know that it must be in $A^c_j$.  We apply equation (\ref{eq:tailBound0}) and we know that it is not likely that $\cup_{j'\in A_j}U_1(Z, j')$ happens:
\begin{align*}
&P(\cup_{j'\in A_j}U_1(Z, j')/\cup_{j'\in A_j}U_2(Z, j') )\\
\leq & |A_j|\max_{j'\in A_j}P(Z_{j'}\leq Z_{j_0}+\sqrt{n}( {\rm Err}_{j_0} - {\rm Err}_{j'}))\rightarrow 0
\end{align*}
(2)If $A_j\cap J_{good}\neq \emptyset$, and $J_{good}/A_j = \emptyset$. We apply equation (\ref{eq:tailBound0}) and we know  that  it is not likely that  $\cup_{j'\in A_j}U_2(Z, j')$ does not happen:
\begin{align*}
&P(\cup_{j'\in A_j}U_1(Z, j')/\cup_{j'\in A_j}U_2(Z, j') )\\
\leq & m|A_j|\max_{j'\in A_j}\max_{l\notin A_j}P(Z_{j'}\geq Z_{l}+\sqrt{\frac{n}{K}}({\rm Err}_{l} - {\rm Err}_{j'}))\rightarrow 0
\end{align*}
(3)If $A_j\cap J_{good}\neq \emptyset$, and $J_{good}/A_j \neq \emptyset$. In this case, both $A_j$ and $A_{j'}$ has good model indexes. 

Let $B_{k}:=\{\min_{l\notin A_j} Z_l + \sqrt{n}({\rm Err}_l - {\rm Err}_{k})\leq Z_k < \min_{l\notin A_j} Z_{l} + \sqrt{\frac{n}{K}}({\rm Err}_l - {\rm Err}_{k})\}$ and  $C_{k,l }:=\{ Z_l + \sqrt{n}({\rm Err}_l - {\rm Err}_{k})\leq Z_k <Z_{l} + \sqrt{\frac{n}{K}}({\rm Err}_l - {\rm Err}_{k})\}$. We know that we need only to consider the indexes in $A_j$ which also belong to $J_{good}$ . Applying equation (\ref{eq:tailBound0}):
\begin{align*}
\lim_{N\rightarrow \infty} \sum_{j'\in A_j, j'\notin J_{good}}P(B_{j'})  = 0
\end{align*}
We have
\begin{align*}
&\lim_{n\rightarrow \infty}P(\cup_{j'\in A_j}U_1(Z, j')/\cup_{j'\in A_j}U_2(Z, j') )\\
\leq &\lim_{n\rightarrow \infty} \sum_{j'\in A_j\cap J_{good}}P(B_{j'}) 
\end{align*}
For  every $j'\in A_j\cap J_{good}$, we have
\begin{align*}
&P(B_{j'})\\
\leq& P(C_{j', l'}; l' = \arg\min_{l\notin A_j} Z_l + \sqrt{n}({\rm Err}_l - {\rm Err}_{j'}), l'\in J_{good})\\
+ &P( l' = \arg\min_{l\notin A_j} Z_l +\sqrt{n}\rm{Err}_l, l'\in J_{bad})
\end{align*}
The first term goes to 0 because when both $j'$ and $l'$ are from $J_{good}$, $\sqrt{n}({\rm Err}_{l'}- {\rm Err}_{j'})$ goes to 0 by Assumption \ref{ass:ass1}. The probability of event $C_{j', l'}$ is then the integral of a normal density with inverval goes to 0, and it also goes to 0 itself. The second term goes to 0 as a result of equation (\ref{eq:tailBound0}). As a consequence
$$
\forall j'\in A_j\cap J_{good}, \;\;\lim_{n\rightarrow 0}P(B_{j'}) = 0
$$
and
$$
\lim_{n\rightarrow \infty}P(\cup_{j'\in A_j}U_1(Z, j')/\cup_{j'\in A_j}U_2(Z, j') ) = 0
$$
Similarly,  we have
$$\lim_{n\rightarrow \infty }P(\cup_{j'\in A_j}U_2(Z, j')/\cup_{j'\in A_j}U_1(Z, j')) =0$$
We combine them together to get the desired result
\[
|\sqrt{n}\Delta_{out}(n) - \sqrt{\frac{n}{K}}\Delta_{out}(\frac{n}{K})|\rightarrow 0
\]
which directly lead to
\[
\sqrt{n}(E[\widehat{Q}(\mathcal{R})] - \rm{Err}(\mathcal{R}))\rightarrow 0
\]

\end{proof}

\bibliographystyle{agsm}
\bibliography{debiasCV.bib}
\end{document}